\documentclass[11pt,letterpaper]{article}

\usepackage[margin=1in]{geometry}
\usepackage{latexsym,graphicx,amssymb}
\usepackage{amsmath}
\usepackage{amsthm, amsfonts}
\usepackage{mathtools, thmtools}
\usepackage{bbm}
\usepackage{float}
\usepackage{xspace}
\usepackage{paralist}
\usepackage{enumerate,multicol}
\usepackage{cases}
\usepackage{caption}
\usepackage{graphicx}
\usepackage{tikz,pgfmath}
\usetikzlibrary{matrix,positioning,quotes}
\usetikzlibrary{shapes.multipart}
\usetikzlibrary{calc}
\usetikzlibrary{arrows,decorations.markings}
\usetikzlibrary{matrix}
\usepackage{url}
\usepackage[ruled,linesnumbered]{algorithm2e}
\usepackage{cleveref}

\usepackage{ulem}
\usepackage[noend]{algpseudocode}
\usepackage{xcolor}
\usepackage{multirow}
\usepackage{graphicx}
\usepackage{subcaption}
\usepackage{varwidth}
\usepackage{wrapfig}
\usepackage{enumitem}
\usepackage{float}
\usepackage{bbm,bm}
\usepackage{tcolorbox}

\numberwithin{figure}{section}
\numberwithin{equation}{section}
\newtheorem{definition}{Definition}[section]

\newtheorem{corollary}{Corollary}[section]

\newtheorem{lemma}{Lemma}[section]
\newtheorem*{lemma*}{Lemma}

\newcommand{\eps}{\varepsilon}
\newcommand{\eqdef}{\overset{\mathrm{def}}{=\mathrel{\mkern-3mu}=}}

\newcommand{\ind}[1]{\mathbbm{1}[ #1  ]} %
\newcommand{\apbalance}{\textup{approximate Perturbed-Balance}\xspace}
\newcommand{\dx}{\textup{d}x}

\newcommand{\bk}[1]{\left({#1}\right)}
\newcommand{\Bk}[1]{\left[{#1}\right]}
\newcommand{\BK}[1]{\left\{{#1}\right\}}
\newcommand{\abs}[1]{\left|{#1}\right|}

\newcommand{\E}{\mathop{\mathbb{E}}}
\renewcommand{\Pr}{\mathop{\mathrm{Pr}}}

\renewcommand{\d}{\mathrm{d}}

\newcommand\numberthis{\addtocounter{equation}{1}\tag{\theequation}}

\newcommand{\opt}{\textup{OPT}}
\newcommand{\alg}{\textup{ALG}}
\newcommand{\rank}{\textup{Ranking}\xspace}
\newcommand{\balance}{\textup{Balance}\xspace}
\newcommand{\prank}{\textup{Perturbed-Ranking}\xspace}
\newcommand{\pbalance}{\textup{Perturbed-Balance}\xspace}

\definecolor{applegreen}{rgb}{0.55, 0.71, 0.0}

\title{On the Perturbation Function of Ranking and Balance for Weighted Online Bipartite Matching}
\author{
 Jingxun Liang \thanks{IIIS, Tsinghua University. Email: \texttt{liangjx20@mails.tsinghua.edu.cn}}
 \and
 Zhihao Gavin Tang\thanks{ITCS, Shanghai University of Finance and Economics. Email: \texttt{tang.zhihao@mail.shufe.edu.cn}}
 \and
 Yixuan Even Xu \thanks{IIIS, Tsinghua University. Email: \texttt{xuyx20@mails.tsinghua.edu.cn}}
 \and
 Yuhao Zhang\thanks{Shanghai Jiao Tong University. Email: \texttt{zhang\_yuhao@sjtu.edu.cn}}
 \and 
 Renfei Zhou \thanks{IIIS, Tsinghua University. Email: \texttt{zhourf20@mails.tsinghua.edu.cn}}}
 
\date{}

\begin{document}

\maketitle
\thispagestyle{empty}

\begin{abstract}
Ranking and Balance are arguably the two most important algorithms in the online matching literature. They achieve the same optimal competitive ratio of $1-1/e$ for the integral version and fractional version of online bipartite matching by Karp, Vazirani, and Vazirani (STOC 1990) respectively. 
The two algorithms have been generalized to weighted online bipartite matching problems, including vertex-weighted online bipartite matching and AdWords, by utilizing a perturbation function. 
The canonical choice of the perturbation function is $f(x)=1-e^{x-1}$ as it leads to the optimal competitive ratio of $1-1/e$ in both settings.

We advance the understanding of the weighted generalizations of Ranking and Balance in this paper, with a focus on studying the effect of different perturbation functions. First, we prove that the canonical perturbation function is the \emph{unique} optimal perturbation function for vertex-weighted online bipartite matching. In stark contrast, all perturbation functions achieve the optimal competitive ratio of $1-1/e$ in the unweighted setting. 
Second, we prove that the generalization of Ranking to AdWords with unknown budgets using the canonical perturbation function is at most $0.624$ competitive, refuting a conjecture of Vazirani (2021). More generally, as an application of the first result, we prove that no perturbation function leads to the prominent competitive ratio of $1-1/e$ by establishing an upper bound of $1-1/e-0.0003$. 
Finally, we propose the online budget-additive welfare maximization problem that is intermediate between AdWords and AdWords with unknown budgets, and we design an optimal $1-1/e$ competitive algorithm by generalizing Balance.

\end{abstract}

\clearpage
\setcounter{page}{1}

\section{Introduction}
\label{sec:intro}

Online bipartite matching has been extensively studied since the seminal work of Karp, Vazirani, and Vazirani~\cite{stoc/KarpVV90}. Two remarkable algorithms, \rank of Karp et al.~\cite{stoc/KarpVV90} and \balance of Kalyanasundaram and Pruhs~\cite{tcs/KalyanasundaramP00}, achieve the same optimal competitive ratio of $1-1/e$ for the integral (randomized) and fractional (deterministic) version of the problem respectively.

\subparagraph*{AdWords and Vertex-weighted.} 
Motivated by the display of advertising on the Internet, Mehta et al.~\cite{jacm/MehtaSVV07} generalized the online bipartite matching problem so that it allows weighted graphs.
Consider an underlying bipartite graph $G = (L\cup R, E)$ with $L,R$ being offline and online vertices. Each vertex $u\in L$ is associated with a budget $B_u$, and each edge $(u,v)\in E$ is associated with a bid $w_{uv}$. The offline vertices and their corresponding budgets are known in advance. 
The online vertices arrive one at a time, with their incident edges and associated bids, and have to be matched immediately and irrevocably to some $u\in L$. An offline vertex $u\in L$ can be matched to multiple online vertices. Let $S_u$ be the set of online vertices matched to $u$ and then, the revenue of $u$ equals $\min\{B_u,\sum_{v \in S_u} w_{uv}\}$, that is, the revenue cannot exceed the budget. 
The goal is to maximize the total revenue and to compete against the optimal revenue of an offline algorithm that knows the whole graph.

Mehta et al.~\cite{jacm/MehtaSVV07} established an optimal $\left(1-1/e\right)$-competitive algorithm for the fractional version\footnote{The paper states their result with the small-bid assumption, i.e., $\gamma = \max_{u,v} w_{uv}/B_u$ is small. We consider the fractional AdWords problem, corresponding to the case when $\gamma \to 0$. See our discussion in Section~\ref{sec:prelim}.} of the problem by generalizing \balance to \pbalance. 
Later, Aggarwal et al.~\cite{soda/AggarwalGKM11} considered a restricted setting of AdWords, called vertex-weighted online bipartite matching, in which all edges incident to $u$ have the same weight of $w_u = B_u$. They generalized \rank to \prank and obtained the same $1-1/e$ competitive ratio for the integral version of the problem.

The two generalizations are both greedy-based algorithms with a careful perturbation of the weights. Specifically, upon the arrival of each vertex $v$, the algorithm matches it with the offline vertex $u$ with the maximum perturbed weight $f(x_u) \cdot w_{uv}$, among those neighbors whose budgets have not yet been exhausted. Here, $x_u$ corresponds to the random rank of $u$ in \prank and the fraction of budget spent so far in \pbalance. 
The canonical choice of the perturbation function is $f(x) = 1-e^{x-1}$, applied by Mehta et al.~\cite{jacm/MehtaSVV07} and Aggarwal et al.~\cite{soda/AggarwalGKM11}.
Notably, Devanur, Jain, and Kleinberg~\cite{soda/DevanurJK13} provided a unified primal-dual analysis for \prank and \pbalance, in which the perturbation function plays a critical role even for the unweighted online bipartite matching problem.

Despite the above successful generalizations of \rank and \balance from unweighted to weighted graphs, we lack an understanding of the extra difficulty introduced by weighted graphs.
In this paper, we revisit the two classic algorithms and focus on the perturbation function.
Notice that for unweighted graphs, \prank (resp. \pbalance) with an arbitrary perturbation function degenerates to the same \rank (resp. \balance) algorithm and achieves the optimal competitive ratio of $1-1/e$. 
We examine the importance of perturbation functions by studying the performance of \prank and \pbalance on weighted graphs with an arbitrary perturbation function.

Our first result confirms the importance of the perturbation function, proving that the canonical perturbation function $f(x)=1-e^{x-1}$ is the unique optimal perturbation function for vertex-weighted online bipartite matching.

\begin{tcolorbox}[frame empty]
\begin{restatable}{theorem}{thmvertexweight}
    \label{thm:vertexweight}
	The perturbation function $f(x) = 1-e^{x-1}$ is unique (up to a scale factor) for \prank and \pbalance to achieve the optimal competitive ratio of $1-1/e$ for vertex-weighted online bipartite matching. 
\end{restatable}
\end{tcolorbox}
It is surprising that this question has been overlooked by the online matching community. We introduce a new family of hard instances that heavily exploits the power of weighted graphs. Noticeably, prior to our work, the only $1-1/e$ impossibility result by Karp et al.~\cite{stoc/KarpVV90} is established for unweighted graphs.
This advanced understanding is also the starting point for us to explore the limitation of \prank in a more general model, i.e., AdWords with unknown budgets.

\subparagraph*{AdWords with Unknown Budgets.} 
A major open question left by Mehta et al.~\cite{jacm/MehtaSVV07} is the competitive ratio of \prank for the fractional AdWords problem.
In addition to obvious theoretical interests, the \prank algorithm has a merit of budget-obliviousness, as pointed out by Vazirani~\cite{corr/Vazirani21} and Udwani~\cite{corr/Udwani21}. I.e., the algorithm does not need to know the budget of each vertex, in contrast to the \pbalance algorithm.
Formally, consider the setting of AdWords with unknown budgets: the algorithm has no prior knowledge of the budgets and only learns the budget of each vertex $u$ when the total revenue of $u$ first exceeds its budget. Observe that \pbalance is not applicable in this setting, since its decision at each step depends on the fraction of budget spent on each offline vertex.

\prank is the only known algorithm for AdWords with unknown budgets so far. Using the canonical perturbation function $f(x)=1-e^{x-1}$, Vazirani~\cite{corr/Vazirani21} proved \prank is $\left(1-1/e\right)$-competitive assuming a no-surpassing property. Udwani~\cite{corr/Udwani21} proved that the algorithm is $0.508$-competitive in the general case and is $0.522$-competitive with a different perturbation function $f(x) = 1 - e^{1.15(x-1)}$.

It is natural to ask if other perturbation functions can lead to a better competitive ratio, or even $1 - 1/e$. In this paper, we give a limitation of \emph{all} perturbation functions, showing a separation between vertex-weighted online bipartite matching and AdWords on the performance of \prank.

We first show that \prank with the canonical perturbation function can only achieve a competitive ratio of at most $0.624 < 1-1/e$. %
Then, together with the family of instances we constructed in the proof of \Cref{thm:vertexweight}, we manage to prove that any perturbation function cannot lead to the prominent competitive ratio of $1-1/e$:

\begin{tcolorbox}[frame empty]
\begin{restatable}{theorem}{thmgeneral}
\label{thm:general}
The competitive ratio of \prank algorithm with any perturbation function $f(x)$ on fractional AdWords is at most $1-1/e-0.0003$. In particular, using the canonical function $f(x) = 1-e^{x-1}$, the competitive ratio is at most $0.624$.
\end{restatable}
\end{tcolorbox}
Our result refutes the conjecture of Vazirani~\cite{corr/Vazirani21} that \prank is $1-1/e$ competitive. Moreover, our construction is clean and simple, suggesting that the no-surpassing assumption might be too strong to hold in reality. Such result leads to the conjecture that there is no $1 - 1/e$ competitive algorithm for AdWords with unknown budgets.

\subparagraph*{Remark} Very recently, independently by our work, Udwani~\cite{corr/Udwani21} updated his paper and proved that a specific perturbation function family $f(x) = 1 - e^{\beta (x - 1)}$ is at most $0.624$-competitive for any $\beta > 0$. Our result provides a stronger observation by another approach that shows \emph{all} perturbation functions cannot achieve the competitive ratio of $1-1/e$.

\subparagraph*{Online Budget-additive Welfare Maximization.}
The above upper bound of the competitive ratio for \prank suggests that AdWords with unknown budgets should be strictly harder than AdWords, in terms of the worse-case competitive ratio. Unfortunately, our current construction is specific to the \prank algorithm and does not serve as a problem hardness.

Inspired by the online submodular welfare maximization problem (that we discuss below in the related work section), we consider a variant which we call online budget-additive welfare maximization problem, that lies in between the original AdWords and AdWords with unknown budgets.
Specifically, we assume that 1) the algorithm has no information of the budgets at the beginning, 2) at each step, the algorithm can query for each vertex $u$, the value of $w_u(S) = \min\{B_u, \sum_{v \in S} w_{uv}\}$ for any subset $S$ of arrived vertices. Notice that AdWords with unknown budgets can be interpreted in a similar way, except that the algorithm can only query those sets $S$ that are subsets of $S(u)\cup \{v\}$ where $S(u)$ is the set of current matched vertices to $u$.

Our final result is an optimal algorithm for the fractional version of the problem. 
We hope it would shed some light on designing online algorithms beyond \prank in the AdWords with unknown budgets setting and designing algorithms beyond greedy (with unrestricted computational power) in the online submodular welfare maximization setting. 

\begin{tcolorbox}[frame empty]
\begin{restatable}{theorem}{thmbudgetadditive}
	\label{thm:budgetadditive}
	There exists a fractional algorithm that achieves the competitive ratio of $(1-1/e)$ for the Online Budget-Additive Welfare Maximization problem.
\end{restatable}
\end{tcolorbox}

\subsection{Related Work}
\label{sec:related}

The seminal work of Karp et al.~\cite{stoc/KarpVV90} studied the unweighted and one-sided online bipartite matching model and proposed the optimal $(1-1/e)$-competitive algorithm: \rank. The analysis of \rank has been refined and simplified by a series of works \cite{sigact/BenjaminM08,soda/DevanurJK13,soda/GoelM08,sosa/EdenFFS21}. Kalyanasundaram and Pruhs studied the $b$-matching model and designed \balance (fractional) that also achieves the competitive ratio of $1-1/e$. The model has been generalized to many weighted variants, e.g., vertex-weighted~\cite{soda/AggarwalGKM11,DBLP:conf/stoc/HuangS21,talg/HuangTWZ19,DBLP:conf/wine/JinW21}, edge-weighted~\cite{DBLP:conf/focs/BlancC21,focs/FahrbachHTZ20,DBLP:conf/focs/GaoHHNYZ21}, and  AdWords~\cite{DBLP:conf/sigecom/DevenurH09,jacm/MehtaSVV07,focs/HuangZZ20}.
Besides the aforementioned generalization of \rank and \balance to the vertex-weighted and AdWords settings, Huang et al.~\cite{talg/HuangTWZ19} generalized \rank to the vertex-weighted setting with random arrival order, by utilizing a two-dimensional perturbation function. They achieved a competitive ratio of $0.653$ that is subsequently improved to $0.662$ by Jin and Williamson~\cite{DBLP:conf/wine/JinW21}. 
Another line of work adapts \rank and \balance to other matching problems, including online bipartite matching with random arrivals~\cite{stoc/KarandeMT11,stoc/MahdianY11}, oblivious matching~\cite{siamcomp/ChanCWZ18,stoc/TangWZ20} and fully online matching~\cite{jacm/HuangKTWZZ20,soda/HuangPTTWZ19,focs/HuangTWZ20}.

The most general extension of online bipartite matching is the online submodular welfare maximization problem. It captures most of the weighted variants of online bipartite matching discussed above. 
In this problem, a set of $n$ offline vertices are given, each associated with a monotone submodular function $w_u$. Upon the arrival of an online vertex, it must be assigned to one of the offline vertices and the goal is to maximized the welfare $\sum_{u} w_u(S_u)$, where $S_u$ is the set of vertices received by $u$.
The algorithm is assumed to have value oracles for the functions. I.e., an algorithm can query the value of $w_u(S)$ for an arbitrary subset $S$ of arrived online vertices.
Kapralov et al.~\cite{DBLP:conf/soda/KapralovPV13} proved that the $0.5$-competitive greedy algorithm is optimal with restricted computational powers. For the unknown i.i.d. setting, they provided an optimal $(1-1/e)$-competitive algorithm. In the random arrival model, Korula et al.~\cite{DBLP:journals/siamcomp/KorulaMZ18} proved that greedy is at least $0.5052$-competitive, and the ratio is improved to $0.5096$ by Buchbinder et al.~\cite{DBLP:journals/mp/BuchbinderFFG20}.
Our budget-additive welfare maximization problem is a special case of the submodular welfare maximization problem where every $w_u$ is a budget-additive function and admits an $(1-1/e)$-competitive algorithm.

Moreover, the AdWords with unknown budgets problem suggests us to study a more restricted oracle access for online submodular welfare maximization. We call it \emph{marginal oracle}, that on the arrival of an online vertex $v$, the algorithm can only query the value of $w_u(S)$ for $S \subseteq S_u(v) \cup \{v\}$, where $S_u(v)$ is the current matched vertex set to $u$.
Based on our discoveries in this paper, we make the following three conjectures for future work:
\begin{itemize}
    \item Online submodular welfare maximization with marginal oracles does \emph{not} admit a $1-1/e$ competitive algorithm.
    \item AdWords with unknown budgets does \emph{not} admit a $1-1/e$ competitive algorithm. 
    \item Online submodular welfare maximization admits a $0.5+\Omega(1)$ competitive algorithm.
\end{itemize}
All the three conjectures assume unlimited computational powers so that the third conjecture does not violate the impossibility result of \cite{DBLP:conf/soda/KapralovPV13}.
Observe that if the second conjecture holds, it automatically confirms the first conjecture, and implies a price of  budget-obliviousness for AdWords.

\section{Preliminaries}
\label{sec:prelim}

We first give the formal definitions of \prank and \pbalance for the vertex-weighted online bipartite matching problem and then discuss their extensions to the fractional AdWords problem. Both algorithms depends on a perturbation function.

\begin{definition}[Perturbation Function]
A perturbation function is a non-increasing and right continuous function $f(x):[0,1] \to [0,1]$.
\end{definition}

\subsection{Vertex-weighted}
Given a perturbation function $f$, the two algorithms are defined as below. Observe that \prank is a randomized algorithm and \pbalance is a deterministic algorithm.
\begin{definition}[\prank~\cite{soda/AggarwalGKM11}]
Sample a rank $y_u$ for each offline vertex $u\in L$ independently from a uniform distribution on $[0,1]$. On the arrival of an online vertex $v$, we match $v$ to the unmatched neighbor $u$ with maximum perturbed weight $f(y_u)\cdot w_{u}$.
\end{definition}

\begin{definition}[\pbalance]
On the arrival time of an online vertex $v$, we continuously match $v$ to the offline neighbor $u$ with maximum perturbed weight $f(x_u)\cdot w_{u}$, where $x_u$ is the current matched portion of $u$. 
\end{definition}

We remark that in the context of \prank, a perturbation function can be interpreted as an alternative representation of a $[0,1]$-bounded random variable, in which $f(x)$ corresponds to the value of a quantile $x$. Moreover, the right continuity is necessary for the \pbalance algorithm to be well-defined.

\subsection{AdWords}
In Section~\ref{sec:adwords} (and Section~\ref{sec:budget_additive}), we shall work on the fractional version of AdWords (and its variant) that is (slightly) more relaxed than the AdWords problem with small bid assumption. See e.g. Udwani~\cite{corr/Udwani21} for a more detailed discussion.
\subparagraph*{Fractional AdWords.} The fractional AdWords problem allows each edge $(u,v)$ to be fractionally matched by an amount of $x_{uv} \in [0,1]$, as long as the total matched portion of each online vertex $v$ is no more than a unit, i.e., $\sum_{u} x_{uv} \le 1$. 

Consider the following generalizations of \prank and \pbalance for the fractional AdWords problem:
\begin{definition}[\prank~\cite{corr/Vazirani21,corr/Udwani21}]
Sample a rank $y_u$ for each offline vertex $u\in L$ independently from a uniform distribution on $[0,1]$. On the arrival of an online vertex $v$, we \emph{continuously} match $v$ to the neighbor $u$ with maximum perturbed weight $f(y_u)\cdot w_{uv}$, among those neighbors whose budgets have not been exhausted yet.
\end{definition}

\begin{definition}[\pbalance (a.k.a. MSVV~\cite{jacm/MehtaSVV07})]
On the arrival time of an online vertex $v$, we continuously match $v$ to the offline neighbor $u$ with maximum perturbed weight $f(x_u/B_u)\cdot w_{uv}$, where $x_u$ is the current used budget of $u$. 
\end{definition}

\section{Vertex-Weighted}
\label{sec:vertex_weighted}
In this section, we consider vertex-weighted online bipartite matching. We prove that to achieve the optimal competitive ratio of $\Gamma=1-1/e$, the canonical choice of the perturbation function $f(x)=1-e^{x-1}$ is unique (up to a scale factor).

Our result holds for both \prank and \pbalance. Indeed, we establish a dominance of \pbalance over \prank in terms of worst-case competitive ratio.
\begin{restatable}{lemma}{thmrankingbalance}
	\label{thm:rankingbalance}
	For any perturbation function $f$, the competitive ratio of \prank is at most the competitive ratio of \pbalance
	for vertex-weighted online bipartite matching. %
\end{restatable}
We sketch our proof below and defer the detailed proof to Appendix~\ref{appendix:reduction}.

\begin{proof}[Proof Sketch] 
Given an arbitrary instance $G=(L\cup R, E,w)$, we construct an instance $G'$ so that 
the competitive ratio of \pbalance for $G$ and the competitive ratio of \prank for $G'$ are approximately the same.
\begin{itemize}
\item For each offline vertex $u \in L$, create $N$ duplicates $\{u^{(i)} \}_{i=1}^{N}$ in $G'$ and with weights $w_{u^{(i)}} = w_u$.
\item For each online vertex $v \in R$, create $N$ duplicates $\{v^{(i)}\}_{i=1}^N$ in $G'$ that arrive in a sequence.
\item For each $(u,v) \in E$, let there be a complete bipartite graph between $\{u^{(i)}\}$ and $\{v^{(i)}\}$ in $G'$.
\end{itemize}
Now, we consider the behavior of \prank on $G'$. Intuitively, although the ranks are drawn independently for each offline vertex, the set of $N$ ranks of $\{u^{(i)}\}_{i=1}^{N}$ should be ``close'' to $\left\{\frac{1}{N},\frac{2}{N},\dots,1\right\}$ with high probability when $N$ is sufficiently large. Formally, we prove the following mathematical fact in Appendix~\ref{appendix:concentration}.

  \begin{restatable}{lemma}{lemconcentration}
    \label{cctlemma}
    Let $x_1, \ldots, x_n$ be i.i.d$.$ random variables sampled from $[0, 1]$ uniformly. Let $y_i$ be the $i$-th order statistics
    of $\BK{x_1, \ldots, x_n}$, for $i = 1, \ldots, n$. Then for any parameter $\eps$ with $4n^{-1/4} < \eps < 1$, we have
    \begin{equation}
      \label{eq:cct-1}
      \Pr_{x_1, \ldots, x_n} \Bk{\abs{y_i - \frac{i}{n}} \le \eps,\, \forall i\in [n]} \ge 1 - 2ne^{-\sqrt n / 6}.
    \end{equation}
  \end{restatable}
  For now, we assume the set of ranks of $\{u^{(i)}\}$ is $\{\frac{1}{N}, \frac{2}{N}, \ldots, 1\}$ for each vertex $u$. Then, upon the arrival of $\{v^{(i)}\}$, we are basically running a discretized version of the \pbalance algorithm, with a step size of $\frac{1}{N}$. In Definition~\ref{def:apbalance} of Appendix~\ref{appendix:reduction}, we formalize this intuition by introducing a family of $\eps$-\apbalance algorithms that approaches the behavior of \pbalance when $\eps \to 0$. Moreover, we prove that the \prank algorithm can be interpreted as an $\eps$-\apbalance algorithm when the ranks of $\{u^{(i)}\}$ behave nicely (which is of high probability). 
  Finally, we conclude the proof of the lemma by letting $N$ go to infinity.
\end{proof}

Equipped with the above lemma, we hereafter focus on the easy-to-analyze \pbalance algorithm, since it is deterministic while \prank is randomized.

Our main construction is a family of instances that strongly restricts the shape of the perturbation function. Naturally, our construction is built upon the classical upper triangle graph that gives the tight $1-1/e$ competitive ratio for unweighted online bipartite matching. On the other hand, our construction consists of a few novel gadgets that might be useful for other weighted online matching problems. The following lemma also serves as a 
stepping stone of our result for the AdWords with unknown budget problem in Section~\ref{sec:adwords}.

\begin{lemma}
  \label{lem:vertex_weight}
  If the \pbalance algorithm achieves a competitive ratio of $\Gamma$ for vertex-weighted online bipartite matching, then the perturbation function $f$ satisfies the following:
  \begin{equation}
    \left(\beta+1-e^{\beta-1} - \Gamma \right) \cdot f(\alpha)\ge \left( 1-(1-\Gamma) \cdot e^{\alpha} \right) \cdot \int_0^\beta f(x)\d  x, \quad \forall \alpha,\beta \in [0,1].
    \label{eq:variables_separable}
  \end{equation}
\end{lemma}

We defer its proof till the end of the section and proceed by first proving our main theorem.
\thmvertexweight*
\begin{proof}
  By \Cref{thm:rankingbalance}, it suffices to prove the theorem for \pbalance.
  By \Cref{lem:vertex_weight} with $\Gamma=1-1/e$, the perturbation function $f(x)$ satisfies the following:
  \[
    \frac{f(\alpha)}{1 - e^{\alpha - 1}} \ge \frac{\int_0^\beta f(x)\dx} {\beta - e^{\beta-1} + e^{-1}}, \quad \forall \alpha,\beta\in (0,1).
  \]
  Let $M \eqdef \inf_{\alpha \in (0,1)} \frac{f(\alpha)}{1 - e^{\alpha - 1}}$. We must then have
  \begin{gather}
    f(\alpha) \ge M(1 - e^{\alpha - 1}),\quad \forall \alpha \in [0,1], \label{eq:f_point}\\
    \int_0^\beta f(x) \dx \le M(\beta - e^{\beta-1} + e^{-1}), \quad \forall \beta \in [0,1] \label{eq:f_integral}.
  \end{gather}
  Taking the integral of $f(x)$ and applying \eqref{eq:f_point}, we have
  \[
    \int_0^\beta f(x) \dx \ge M \int_0^\beta (1 - e^{x - 1}) \dx = M(\beta - e^{\beta-1} + e^{-1}).
  \]
  Together with \eqref{eq:f_integral}, we conclude that $f(x) = M (1-e^{x-1})$ for all $x \in [0,1]$ according to the %
  right-continuity of function $f$. 
  Specifically, if there exists an $x^* \in [0,1)$, that $f(x^*) = M (1-e^{x^*-1}) + \eps$ for some $\eps > 0$, then there exists a sufficiently small $\delta > 0$ such that for any $x\in [x^*, x^*+\delta]$, $f(x) \geq M(1 - e^{x - 1}) + \frac{\eps}{2}$. Then we can see by \eqref{eq:f_point} that
  \[
    \int_0^1 f(x) \dx \ge \int_0^1 \left(M (1 - e^{x - 1}) + \ind{x^* \leq x \leq x^*+\delta}\cdot \frac{\eps}{2}\right) \dx = \frac{M}{e} +\frac{\eps \delta}{2},
  \]
  which violates the statement of \eqref{eq:f_integral}, $\int_0^1 f(x)\dx \le \frac{M}{e}$. Therefore, $\forall x\in[0,1)$, $f(x)=M(1 - e^{x - 1})$. Also, for $f(1)$, note that $f(x)$ is decreasing, so $f(1)\leq \lim_{x\to 1^{-}}f(x)=0$. Then $f(1)=0$.

  This shows that  $\forall x\in[0,1]$, $f(x)=M(1 - e^{x - 1})$, concluding the proof of the theorem.
\end{proof}

\subsection{Proof of Lemma~\ref{lem:vertex_weight}}
Fix $\alpha,\beta \in [0,1]$. Let $n,m$ be sufficiently large numbers. Refer to Figure~\ref{fig 1} for our instance.

\begin{figure}[H]
  \centering
  \includegraphics[width = 0.8 \textwidth]{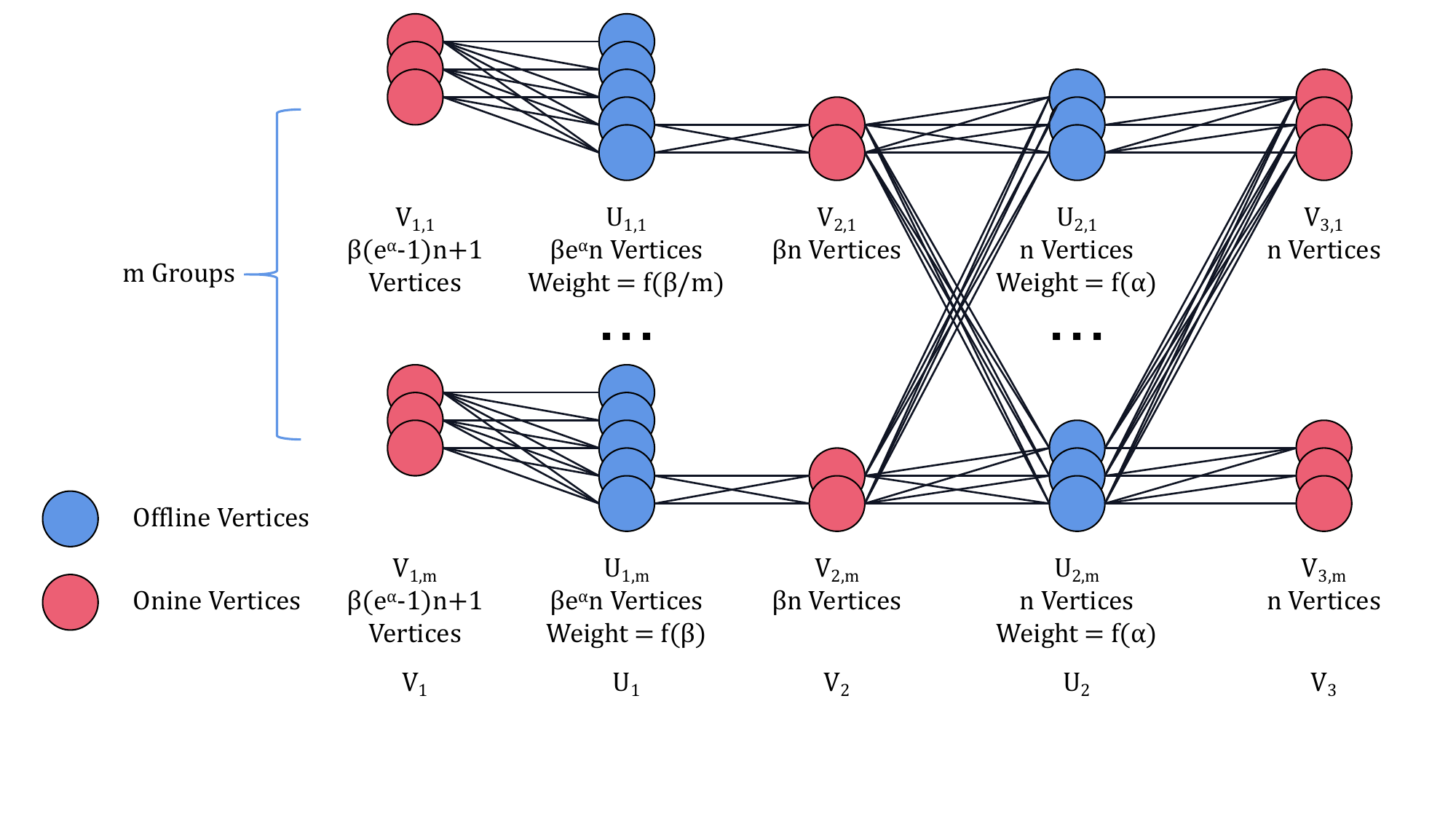}\vspace{-20pt}
  \caption{Instance 1}
  \label{fig 1}
\end{figure}

Our construction consists of $m$ groups of vertices, and each group consists of $5$ parts. 
We use $V_{1,i}, V_{2,i}, V_{3,i}$ to denote the three online parts of group $i$ and $U_{1,i}, U_{2,i}$ to denote the two offline parts of group $i$.
Let $V_j = \cup_{i \in m} V_{j,i}$ for $j \in \{1,2,3\}$ and $U_j = \cup_{i \in m} U_{j,i}$ for $j \in \{1,2\}$.
We first define the vertices of the graph: for each $i \le m$,
\begin{itemize}
\item $U_{1,i}$ consists of $\left(\beta e^\alpha n+1\right)$ offline vertices\footnote{When $\beta e^\alpha n$ is a fraction, let there be $\lceil \beta e^\alpha n \rceil$ vertices. Since we are interested in the case when $n,m$ are sufficiently large, we safely omit the ceiling function for the simplicity of notation. We apply a similar treatment for fractions throughout the paper.} with the same weight of $f(\frac{i}{m}\cdot \beta)$.
\item $U_{2,i}$ consists of $n$ offline vertices with the same weight of $f(\alpha)$. 
\item $V_{1,i}, V_{2,i}, V_{3,i}$ consist of $\beta(e^\alpha-1)n, \beta n, n$ online vertices, respectively.
\end{itemize}
The arrival order of the vertices is the following:
\[
	V_{1,1}\to V_{1,2}\to \cdots\to V_{1,m}\to V_{2,1}\to V_{2,2}\to \cdots\to V_{2,m}\to V_{3,1}\to V_{3,2}\to \cdots\to V_{3,m}.
\]
Next, we define the edges of the graph:
\begin{itemize}
\item $V_{1,i}$ and $U_{1,i}$ are connected as an \textit{upper triangle}. I.e., the $j$-th vertex of $V_{1,i}$ is connected to the $k$-th vertex of $U_{1,i}$ if and only in $k\ge j$.
\item $V_{2,i}$ and the last $\beta n$ vertices in $U_{1,i}$ are \textit{fully connected}.
\item $V_2$ and $U_2$ are \textit{fully connected}.
\item $V_3$ and $U_2$ are connected as a \textit{upper triangle}. I.e., the $j$-th vertex of $V_3$ is connected to the $k$-th vertex of $U_2$ if and only in $k\ge j$.
\end{itemize}

We first calculate the optimum matching of the graph. That is, matching together $(V_3,U_2)$ and $(V_1 \cup V_2,U_1)$. Therefore, we have
\begin{equation}
\label{eq:vertex_opt}
\opt = \underbrace{nm\cdot f(\alpha)}_{(V_3,U_2)}+ \underbrace{\sum_{i=1}^{m}\beta e^\alpha n\cdot f\left(\frac{i}{m}\cdot \beta\right)}_{(V_1 \cup V_2,U_1)} = nm\cdot\left(f(\alpha)+e^\alpha\int_{0}^{\beta}f(x)\d  x + o(1) \right),
\end{equation}
where the second equality holds when $m$ goes to infinity. %

Next, we analyze the the performance of \pbalance. We split the whole instance into three stages, corresponding to the arrivals of $V_1,V_2,V_3$ respectively.
\subparagraph*{First stage ($V_1$).} Upon the arrival of each vertex in $V_1$, it matches uniformly to its neighbours in $U_1$. The behavior of \pbalance is the same for different groups. We analyze the matched portion of the last $\beta n$ vertices of each group after the first stage:  
\[
x_u=\frac{1}{\beta e^\alpha n}+\frac{1}{\beta e^\alpha n-1}+\cdots+\frac{1}{\beta n} = \ln \left( \frac{\beta e^\alpha n}{\beta n} \right) + o(1) = \alpha + o(1), 
\]
where the equality holds when $n$ goes to infinity.

\subparagraph*{Second Stage ($V_2$).} Upon the arrival of each vertex $v$ of $V_{2,i}$, it will be weighing the perturbed weights from $U_{1,i}$ and $U_2$:
\begin{align*}
& f(x_{u_{1}}) \cdot w_{u_{1}} = f(\alpha+o(1)) \cdot f\left(\frac{i}{m}\right), && \text{for } u_1 \in U_{1,i} \cap N(v), \\
\text{and} \quad & f(x_{u_{2}}) \cdot w_{u_{2}} \ge f\left(\frac{i}{m}\right) \cdot f(\alpha), && \text{for } u_2 \in U_2.
\end{align*}
Notice that the perturbed weights from $U_2$ is always larger than the perturbed weights from $U_{1,i}$. We claim that in the second stage, all vertices of $V_2$ would be fully matched to $U_2$ by \pbalance. Thus, at the end of the second stage, all vertices in $U_2$ have matched portion $\beta$.

\subparagraph*{Third stage ($V_3$).} The behavior of the last stage is similar to the behavior of the first stage, except that all vertices in $U_2$ start with a matched portion of $\beta$. 
After the arrival of the $k$-th vertex in $V_3$, the matched portion of its neighbor equals
\[
\beta + \frac{1}{nm} + \frac{1}{nm-1} + \cdots + \frac{1}{nm-k+1} \ge \beta + \ln \left(\frac{nm}{nm-k+1} \right).
\]
Consequently, only the first $(1 - e^{\beta-1}) nm + 1$ vertices from $V_3$ can be matched. For the rest of the online vertices, all their neighbors would be fully matched before their arrivals.

To sum up, we calculate the performance of \pbalance.
\begin{align}
\alg &\le \underbrace{\sum_{i=1}^{m}  \left(\beta (e^\alpha-1) n + 1\right) \cdot f\left(\frac{i}{m}\cdot \beta\right)}_{(V_1,U_1)} + \underbrace{\beta nm\cdot f(\alpha)}_{(V_2,U_2)} + \underbrace{\left((1-e^{\beta-1}) nm +1 \right) \cdot f(\alpha)}_{(V_3,U_2)} \notag \\
&= nm\cdot\left((e^\alpha-1) \int_{0}^{\beta}f(x)\d x+(\beta+1-e^{\beta-1}) \cdot f(\alpha)+o(1)\right).
\label{eq:vertex_alg}
\end{align}

Together with \eqref{eq:vertex_opt} and the assumption that \pbalance is $\Gamma$-competitive, we conclude the proof by letting $n,m \to \infty$:
\begin{align*}
(e^\alpha-1) \int_{0}^{\beta}f(x)\d x+(\beta+1-e^{\beta-1}) \cdot f(\alpha) \ge \Gamma \cdot \left(f(\alpha)+e^\alpha\int_{0}^{\beta}f(x)\d x\right) \\
\iff \left(\beta+1-e^{\beta-1} - \Gamma \right) \cdot f(\alpha)\ge \left( 1-(1-\Gamma) \cdot e^{\alpha} \right) \cdot \int_0^\beta f(x)\d x, \quad \forall \alpha,\beta \in [0,1].
\end{align*}

\section{AdWords with Unknown Budget}
\label{sec:adwords}

\label{sec:unknown}

In this section, we prove \Cref{thm:general}. 
\thmgeneral*

We first construct a hard instance for which \prank with $f(x)=1-e^{x-1}$ only achieves a competitive ratio of $0.624$. 
Recall that the vertex-weighted online bipartite matching problem is a special case of the AdWords problem. 
Together with \Cref{thm:vertexweight}, it should be convincing that \prank (with any perturbation function) cannot achieve the prominent competitive ratio of $1-1/e$. 

Our construction for general perturbation functions has a similar structure as the construction for the canonical perturbation function. 
On the other hand, general perturbation functions introduce extra technical difficulties to our argument that we shall discuss soon.

\subsection{Canonical Perturbation Function $f(x)=1-e^{x-1}$}

We prove the result by the following lemma.

\begin{lemma}
    \label{lem:adwords_special}
    If \prank with perturbation function $f(x)=1-e^{x-1}$ achieves a competitive ratio of $\Gamma$ on AdWords, then
    \begin{equation}
        (1-\Gamma) \cdot f(\alpha) \ge (\Gamma - \alpha) \cdot \int_0^\alpha f(x) \dx + \Gamma \cdot \int_{\alpha}^1 f(x) \dx, \quad \forall \alpha \in [0,1]. \label{eq:adwords_special}
    \end{equation}
\end{lemma}
\begin{proof}
  Fix $\alpha \in [0,1]$. Let $n$ be a sufficiently large number. Refer to \Cref{fig 2} for our instance.

  \begin{figure}[H]
    \centering
    \includegraphics[width = 0.8 \textwidth]{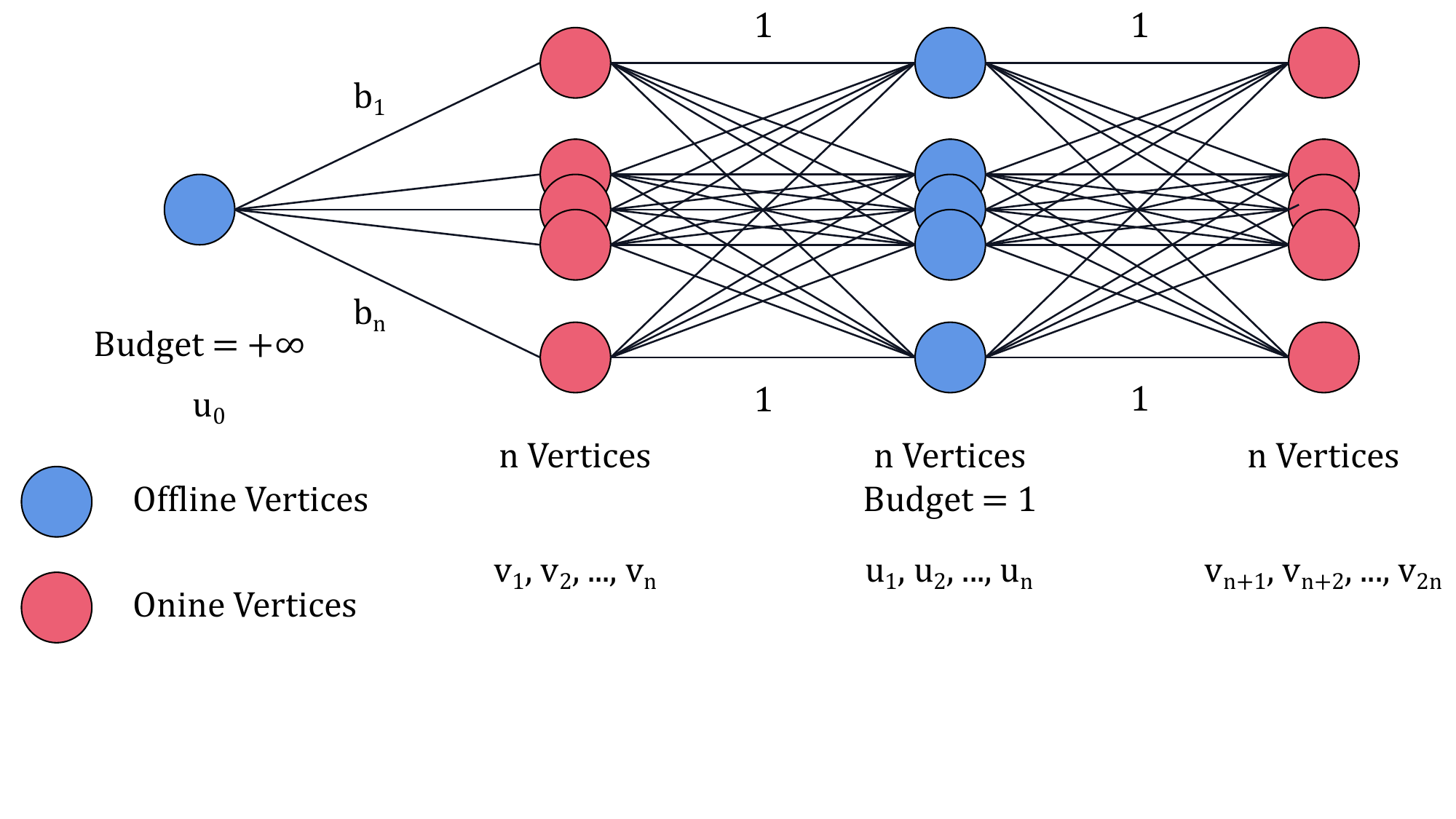}\vspace{-40pt}
    \caption{Instance 2}
    \label{fig 2}
  \end{figure}

  Our construction consists of $n+1$ offline vertices $u_0,u_1,\dots,u_n$ and $2n$ online vertices $v_1,v_2,\dots,v_{2n}$. The budgets of $u_1,u_2,\dots,u_n$ are all $1$ and the budget of $u_0$ is unlimited. The online vertices arrive in ascending order of their indices, i.e. $v_i$ is the $i$-th arriving online vertex. Next, we define the edges of the graph:
\begin{itemize}
    \item $v_1,v_2,\dots,v_{n}$ are connected to $u_0$, with edge weights $b_1,b_2,\dots,b_{n}$ respectively.
    \item $v_1,v_2,\dots,v_{2n}$ are fully connected to $u_1,u_2,\dots,u_{n}$ with edge weights $1$.
\end{itemize}
Before we define the weights, we make an extra assumption to simplify our analysis:
\[
\forall 1\le i \le n, \text{the rank of } u_i \text{ is } y_i=i/n.
\]
Indeed, by \Cref{cctlemma}, we would have that the set of ranks $\{y_1,\ldots,y_n\}$ are ``close'' to $\{\frac{1}{n},\frac{2}{n},\dotsm,1\}$. Moreover, all vertices $u_1,\ldots,u_n$ are symmetric in our graph.
This assumption would significantly simplify our analysis and can be removed by a more conservative choice of the edge weights. We omit the tedious details for simplicity.

Let $b_i=\frac{f(i/n)}{f(\alpha)}$. The offline optimum is to match $v_1,v_2,\dots,v_n$ to $u_0$ and to match $v_{n+1},v_{n+2},\dots,v_{2n}$ to $u_1,u_2,\dots,u_n$, respectively. Consequently,
\begin{align}
\opt & = n+\sum_{i=1}^{n}b_i=n\cdot\left(1+\frac{1}{f(\alpha)}\int_{0}^{1}f(x)\d  x+o(1)\right). \notag
\end{align}

With the assumption, the only randomness of \prank is the rank $y_0$ of $u_0$. 
\subparagraph*{Case 1. ($y_0 \geq \alpha$)} For each online vertex $v_i$, the perturbed weight of $(u_0,v_i)$ is 
\[
b_i\cdot f(y_0)=\frac{f(y_0)}{f(\alpha)}\cdot f\left(\frac{i}{n}\right) \le f \left(\frac{i}{n}\right),
\]
while the perturbed weight of $(u_i,v_i)$ is $f(y_i) = f\left(\frac{i}{n}\right)$. 
Therefore, \prank matches $(u_i,v_i)$ for all $ 1\le i \le n$ and we have $\alg(y_0) =n$.

\subparagraph*{Case 2. ($y_0 < \alpha$)} In this case, some of the $v_1,v_2,\dots,v_n$ would be matched to $u_0$. However, the number of vertices matched to $u_0$ should be no more than $\alpha n$. The reason is as follows. For each online vertex $v_i$, suppose the number of previous vertices matched to $u_0$ is larger than $\alpha n$, then the perturbed weight of $(u_0,v_i)$ is $\frac{f(y_0)}{f(\alpha)}\cdot f\left(\frac{i}{n}\right)$, while the perturbed weight of $(u_{i-\alpha n}, v_i)$ is $f\left(\frac{i}{n}-\alpha\right)$. Notice that $f(x)=1-e^{x-1}$ is a log-concave function. Hence,
\[
\frac{f(y_0)}{f(\alpha)}\cdot f\left(\frac{i}{n}\right) \le \frac{f(0)}{f(\alpha)} \cdot f\left(\frac{i}{n}\right) \leq f\left(\frac{i}{n}-\alpha \right).
\]
In other words, $v_i$ will not match $u_0$. This concludes the proof that the number of vertices matched to $u_0$ is no more than $\alpha n$. Notice that $b_i$'s are non-increasing, we have
\[
\alg(y_0) \leq \sum_{i=1}^{\alpha n} b_i + n = \sum_{i=1}^{\alpha n} \frac{f(i/n)}{f(\alpha)} + n = n \cdot \left( \frac{1}{f(\alpha)}\int_{0}^{\alpha}f(x)\d  x+ 1 + o(1)\right).
\]

Putting the two cases together and assuming that \prank algorithm is $\Gamma$-competitive, we conclude the proof of the lemma.
\begin{align*}
&\Gamma \le \frac{\mathbb E[\alg]}{\opt} = \frac{\alpha \cdot n \cdot \left( \frac{1}{f(\alpha)}\int_{0}^{\alpha}f(x)\dx+ 1 + o(1)\right) + (1-\alpha) \cdot n}{n\cdot\left(1+\frac{1}{f(\alpha)}\int_{0}^{1}f(x)\dx+o(1)\right)} = \frac{\alpha \cdot \int_0^\alpha f(x) \dx + f(\alpha)}{\int_0^1 f(x) \dx + f(\alpha)} \\
&\iff (1-\Gamma) \cdot f(\alpha) \ge (\Gamma - \alpha) \cdot \int_0^\alpha f(x) \dx + \Gamma \cdot \int_{\alpha}^1 f(x)\dx. \qedhere
\end{align*}
\end{proof}

\begin{corollary}
\prank with $f(x)=1-e^{x-1}$ is at most $0.624$-competitive for AdWords.
\end{corollary}
\begin{proof}
Plugging in $\alpha=0.1$ and $f(x)=1-e^{x-1}$ in equation~\eqref{eq:adwords_special},
we have
\[
\Gamma \le \frac{\alpha \cdot \int_0^\alpha f(x) \d x + f(\alpha)}{\int_0^1 f(x) \d x + f(\alpha)} = \frac{0.1 \cdot \left(0.1 - e^{-0.9} + e^{-1}\right) + 1 - e^{-0.9}}{e^{-1} + 1 - e^{-0.9}} < 0.624~.
\qedhere
\]

\end{proof}

\subsection{General Perturbation Functions}
\label{subsec:general-function}
Before we delve into the detailed proof, we explain the technical difficulty introduced by general perturbation functions. Our plan is to generalize \Cref{lem:adwords_special}: if we are able to prove equation~\eqref{eq:adwords_special} for an arbitrary function $f$,  we would then conclude our theorem by combining it with equation~\eqref{eq:variables_separable}.

However, our argument of the second case ($y_0<\alpha$) of {\Cref{lem:adwords_special}} crucially relies on the specific formula of $f(x)$. 
I.e., to upper bound the performance of \prank, we use the fact that $f(0)\cdot f(x) \le f(\alpha)\cdot f(x-\alpha)$.

For a general perturbation function $f(x)$, if we stick to the same property that no more than $\alpha n$ vertices can be matched to $u_0$ when $y_0 < \alpha$, we could achieve it by setting the weights $b_i$ to be smaller. Indeed, if $f\left(\frac{i}{n}-\alpha\right)\geq f(0) \cdot b_i$ holds, the previous analysis can be easily generalized. Hence, a natural attempt is to modify the instance as the following.
\[
  b_i=\left\{\begin{array}{lcl}\frac{1}{f(\alpha)}\cdot f\left(\frac{i}{n}\right)&&{i<\alpha n,}\\\min\left\{\frac{1}{f(\alpha)}\cdot f\left(\frac{i}{n}\right),\frac{1}{f(0)}\cdot f\left(\frac{i}{n}-\alpha\right)\right\}&&{i\geq \alpha n.}\end{array} \right.
\]

Unfortunately, this modification is not strong enough to give a constant strictly smaller than $1-1/e$. The reason is that the function $f$ may have a steep drop in the neighborhood of $0$, which leads to negligible $b_i$'s in the above construction.

On the other hand, the failure of the analysis comes from our coarse and brutal relaxation by establishing a single upper bound for all $y_0 < \alpha$. For instance, if the function steeply drops at some $\beta \in [0,\alpha]$, then we could resolve the issue by considering two cases of $y_0 < \beta$ or $y_0 \in [\beta, \alpha)$.
Formally, we prove the following lemma that is slightly weaker than \eqref{eq:adwords_special}.

\begin{lemma}
    \label{lem:adwords_general}
    If \prank with perturbation function $f(x)$ achieves a competitive ratio of $\Gamma$ on AdWords, then
    \begin{multline}
        (1-\Gamma) \cdot f(\alpha) \ge (\Gamma - \alpha) \cdot \int_0^\alpha f(x)\d x + (\Gamma-\beta) \cdot \int_{\alpha}^1 \min\left\{f(x),\frac{f(\alpha)}{f(\beta)}f(x-\alpha)\right\} \d x,\\
        \forall \alpha,\beta \in [0,1],\alpha\geq \beta.\label{eq:adwords_general}
    \end{multline}
\end{lemma}
\begin{proof}
We apply the same construction as in \Cref{lem:adwords_special} and make the same assumption that $y_i = \frac{i}{n}$ for the simplicity of presentation. We modify the instance by setting the weights $b_i$ differently:
\[
\forall 1\le i\le n, \quad b_i=\left\{\begin{array}{lcl}\frac{1}{f(\alpha)}\cdot f\left(\frac{i}{n}\right)&&{i<\alpha n,}\\\min\left\{\frac{1}{f(\alpha)}\cdot f\left(\frac{i}{n}\right),\frac{1}{f(\beta)}\cdot f\left(\frac{i}{n}-\alpha\right)\right\}&&{i\geq \alpha n.}\end{array} \right.
\]

The optimal solution equals
\begin{equation}
\label{eq:general_opt}
\opt = n+\sum_{i=1}^{n}b_i=n\cdot\left(1+\frac{1}{f(\alpha)}\int_{0}^{\alpha}f(x)\d  x+
    \int_{\alpha}^{1}\min\left\{\frac{f(x)}{f(\alpha)},\frac{f\left(x-\alpha\right)}{f(\beta)} \right\}\d  x+o(1)\right).
\end{equation}

Next, we consider the performance of \prank depending on the value of $y_0$.
\subparagraph*{Case 1. ($y_0<\beta$)} We use $\opt$ as a trivial upper bound of $\alg $, i.e., $\alg(y_0) \leq \opt$.
\subparagraph*{Case 2. ($y_0 \geq \alpha$)} For each online vertex $v_i$, the perturbed weight of $(u_0,v_i)$ is 
\[
b_i\cdot f(y_0) \le \frac{f(y_0)}{f(\alpha)}\cdot f\left(\frac{i}{n}\right) \le f \left(\frac{i}{n}\right),
\]
while the perturbed weight of $(u_i,v_i)$ is $f(y_i) = f\left(\frac{i}{n}\right)$. Therefore, \prank matches $(u_i,v_i)$ for all $ 1\le i \le n$ and we have $\alg(y_0) =n$.

\subparagraph*{Case 3. ($y_0\in [\beta,\alpha)$)} We prove that the number of vertices matched to $u_0$ is no more than $\alpha n$. This can be similarly argued as follows. For each online vertex $i\ (i>\alpha n)$, suppose the number of vertices matched to $u_0$ is already $\alpha n$, the perturbed weight for $u_0$ is $f(y_0) \cdot b_i\leq \frac{f(y_0)}{f(\beta)}\cdot f\left(\frac{i}{n}-\alpha\right)$, while the perturbed weight for $u_{i-\alpha n}$ is $f\left(\frac{i}{n}-\alpha\right)$. So that $f\left(\frac{i}{n}-\alpha\right)\geq f(y_0) \cdot b_i$ and thus $v_i$ will not choose $u_0$ again. Therefore, as the number of vertices matched to $u_0$ is no more than $\alpha n$, we have
\[
\alg \leq n+\sum_{i=1}^{\alpha n}b_i=n\cdot\left(1+\frac{1}{f(\alpha)}\int_{0}^{\alpha}f(x)\d  x+o(1)\right).
\]

Taking expectation over the randomness of $y_0$, we conclude that
\begin{equation}
\label{eq:general_alg}
\mathbb E[\alg ]\leq n\cdot\left(\frac{\alpha-\beta}{f(\alpha)}\int_{0}^{\alpha}f(x)\d  x+1-\beta+o(1)\right)+\beta\cdot \opt .
\end{equation}

Finally, we conclude the proof by plugging in \eqref{eq:general_opt} and \eqref{eq:general_alg} to $\mathbb E[\alg ]\geq \Gamma\cdot \opt $ and letting $n$ goes to infinity.
\end{proof}

Recall that the vertex-weighted online bipartite matching problem is a special case of AdWords. We conclude the proof of \Cref{thm:general} by \Cref{lem:vertex_weight}, \ref{lem:adwords_general} and the following mathematical fact. We provide the proof of the following lemma in \Cref{appendix:mathematical}.
\begin{restatable}{lemma}{lemratio}
  \label{lem:mathematical_fact}
  If a perturbation function $f$ and $\Gamma > 0$ satisfy the following conditions:
  \begin{gather*}
    \left(\alpha+1-e^{\alpha-1} - \Gamma \right) \cdot f(\beta)\ge \left( 1-(1-\Gamma) \cdot e^{\beta} \right) \cdot \int_0^\alpha f(x) \dx, \quad \forall \alpha,\beta \in [0,1], \numberthis \label{ljxeq1} \\
    \begin{aligned}
    (1-\Gamma) f(\alpha) \ge (\Gamma - \alpha)  \int_0^\alpha f(x) \dx + (\Gamma - \beta) \int_{\alpha}^1 \min \left\{f(x),\frac{f(\alpha)}{f(\beta)}f(x-\alpha)\right\} \dx,& \\
    \forall \alpha, \beta \in [0,1],\alpha \ge \beta,&
    \end{aligned}
    \numberthis \label{ljxeq2}
  \end{gather*}
then $\Gamma < 1 - 1/e - 0.0003$.
\end{restatable}

\section{Online Budget-Additive Welfare Maximization}
\label{sec:budget_additive}
This section introduces the $(1-1/e)$-competitive fractional algorithm for online budget-additive welfare maximization. Compared to online submodular welfare maximization, the model restricts the utility function to be a budget-additive function. Compared to the AdWords with unknown budget setting, we can know the information about the budget earlier. (i.e., we know $u$'s budget when the sum of released bid w.r.t. $u$ exceeds its budget.) We study the fractional version as the definition of fractional AdWords.

The fractional algorithm consists of two steps: At first, we introduce an idea called vertex decomposition. We will decompose offline vertices into different stages. Then, we will simulate \pbalance (a.k.a. MSVV~\cite{jacm/MehtaSVV07}) on the virtual graph (after decomposition) we construct. Let us first define the model and introduce the notations.

\subparagraph*{The model}
Given an underlying graph $G=(L\cup R)$, each offline vertex $u\in L$ is associated with an underlying budget $B_u$ and a non-negative marginal gain $w_{uv}$ with each online vertex $v\in R$. The marginal gain $w_{uv}$ is released with $v$'s arrival, and the budget of $u$ is observed by the algorithm at the first time when $\sum_{v\in V'} w_{uv} > B_u$ where $V'$ is the set of already released online vertices. The algorithm can allocate $v$ to some $u$ (i.e., let $x_{uv} = 1$) at $v$'s arrival irrevocably. In the fractional version, we allow $x_{uv}$ to be a fractional value while keeping $\sum_{u\in L} x_{uv} \leq 1$. Finally, the total gain of an offline vertex $u$ equals to $\min\{\sum_{v \in R} w_{uv}x_{uv}, B_u\}$, and we aim to maximize the total welfare of all offline vertices.

Then, we introduce the vertex decomposition. 

\subparagraph*{Vertex Decomposition}
Given an original input instance $(G=(L\cup R),B,w)$, we will construct a new instance $(G'=(L' \cup R), B', w')$ in an online fashion. We sort the online vertices by their releasing time and let $v_i$ denote the $i$-th released online vertex. At the arrival time of an online vertex $v_i$, if the budget of any offline vertex $u$ is not used up, we will create a new stage $u_i$ of $u$ into $L'$. Let $v_m$ be the online vertex whose arrival is the first time we observe $B_u$. We will totally decompose $u$ into $m$ stages in $L'$. The budget and the marginal weight are defined as follows:
\[
  B'_{u_i} = 
  \begin{cases}
  w_{uv_i} & i\leq m-1 \\
  B_u - \sum_{j=1}^{m-1} w_{uv_j} & i = m
  \end{cases}
  \qquad
  w'_{u_jv_i} =
  \begin{cases}
    w_{uv_i} & j \leq i \\
    0        & j>i.
  \end{cases}
\]
Remark that we only have non-zero marginal gains between $v_i$ and already constructed stages of offline vertices (i.e., $u_j$ such that $j\leq i$).

\subparagraph*{The Optimal Fractional Algorithm}

Then, we are ready to define the fractional algorithm. At the arrival time of an online vertex $v$, let $x_{u_i}$ be the allocated portion of stage $i$ of $u$. We check all existed and unsaturated stages of all offline vertices and continuously allocate $v$ to one stage of an offline vertex with the maximized perturbed weight, i.e., $\arg\min_{u_i} f(x_{u_i}/B_{u_i}) \cdot w_{u_iv}$. Whenever a stage of $u$ has been chosen, we will simultaneously increase $x_{uv}$ and $x_{u_i}$, i.e., we virtually match $v$ with $u_i$ and actually match $v$ with $u$.

To prove the competitive ratio, we first show that the two instances share the same offline optimal fractional solution.
\begin{lemma}
  \label{lem:budet-additive-opt}
  Let $\textup{OPT}(G)$ and $\textup{OPT}(G')$ be the offline optimal fractional solution of $(G,B,w)$ and $(G',B',w')$ respectively. We have
  \[
    \textup{OPT}(G) = \textup{OPT}(G').
  \]
\end{lemma}
\begin{proof}
  For any allocation $x'$ of $(G',B',w')$, fixing an offline vertex $u$ with $m$ stages, we let $x_{uv} = \sum_{i=1}^m x'_{u'v}$. Notice that, because $B_u = \sum_{i=1}^{m} B_{u_i}$ and $w_{uv} \geq w_{u_iv}$,
  \[
    \sum_{i=1}^{m} \min \left\{\sum_{v\in V} w'_{u_iv}x'_{u_iv},B'_{u_i}\right\} \leq \min \left\{\sum_{v\in V} w_{uv}x_{uv},B_u\right\}.
  \]
  Thus, we have $\textup{OPT}(G') \leq \textup{OPT}(G)$.

  On the other hand, for any allocation $x$ of $(G,B,w)$, we also consider an offline vertex $u$ with $m$ stages. For each offline vertex $u$, we simulate the gain of $x$ in $(G,B,w)$ by $x'$ in $(G',B',w')$. For any online vertex $v_i$ with $i<m$, we construct $x'_{u_iv_i} = x_{uv_i}$. Notice that, none of $u_i$ exceeds its budget now, and we totally get the gain and use the budget of $\sum_{i=1}^{m-1} w'_{u_iv_i} x'_{u_iv_i}$, which is the same as that in $x$ (i.e., $\sum_{i=1}^{m-1} w_{uv_i} x_{uv_i}$). For all the other online vertices $v_i$ with $i\geq m$, they are connected to all $u_1$ to $u_m$ with the same marginal gain $w_{uv_i}$. We can simulate $x_{uv_i}$ by keep finding an unsaturated stage $j$ of $u$ and increase $x'_{u_jv_i}$ until $u$'s budget is totally used up. Therefore, we have:
  \[
    \sum_{i=1}^{m} \min \left\{\sum_{v\in V} w'_{u_iv}x'_{u_iv},B'_{u_i}\right\} = \min \left\{\sum_{v\in V} w_{uv}x_{uv},B_u\right\}.
  \]
  It concludes $\textup{OPT}(G') \geq \textup{OPT}(G)$, so the lemma is proved.
\end{proof}

By the result of Mehta et al \cite{jacm/MehtaSVV07}, we have:
\begin{lemma}[\cite{jacm/MehtaSVV07}]
  If we run \textup{MSVV} on $(G',B',w')$, we have
  \[
    \textup{MSVV}(G') \geq (1-1/e)\textup{OPT}(G').
  \]
\end{lemma}
Then, we prove that running the algorithm on $(G,B,w)$ is the same as running MSVV on $(G',B',w')$, so that we have the following corollary.
\begin{corollary}
  Let \textup{ALG} be the gain collected by the fractional online algorithm,
  \[
    \textup{ALG}(G) \geq (1-1/e)\textup{OPT}(G).
  \]
\end{corollary}
\begin{proof}
  Although vertices in $L'$ are constructed over time in our instance, we observe that online vertices only have non-zero marginal gains with already constructed offline stages in $L'$. Thus, MSVV keeps the same performance on $(G',B',w')$ whether $L'$ is given upfront or over time in such an online fashion. Moreover, as the budget of $u$ equals to the sum of budget from $u_1$ to $u_m$ and $w_{u_iv} \leq w_{uv}$, for every fixed offline vertex $u$, the algorithm's gain of $u$ on $(G,B,w)$ is always larger than or equal to MSVV's gain on $u_1$ to $u_m$ in $(G',B',w')$.
  By combining \Cref{lem:budet-additive-opt}, we have
  \[
    \textup{ALG}(G) \geq \textup{MSVV}(G') \geq (1-1/e)\textup{OPT}(G') = (1-1/e)\textup{OPT}(G). \qedhere
  \]
\end{proof}
Finally, the corollary concludes \Cref{thm:budgetadditive}.

\bibliographystyle{plain}
\bibliography{matching.bib}

\newpage
\appendix

\section{Formal proof of \Cref{thm:rankingbalance}}
\subsection{Concentration Lemma}
\label{appendix:concentration}
\lemconcentration*
\begin{proof}
  We first prove that for any fixed $i$, if $\eps > 4 n^{-1/4}$, then
  \begin{equation}
    \label{eq:cct-2}
    \Pr_{x_1, \ldots, x_n} \Bk{y_i > \frac{i}{n + 1} + \frac{\eps}{2}} \le e^{-\sqrt n / 6}.
  \end{equation}
  Define $p = \frac{i}{n + 1} + \frac{\eps}{2}$. If $p \ge 1$, the event $y_i > p$ never happens, so \eqref{eq:cct-2} is trivially correct. Otherwise, let $z_j = [x_j \le p]$, the event $y_i > p$ is equivalent to $\sum_{j=1}^n z_j < i$. Notice that $z_j$ are i.i.d$.$ Bernoulli random variables, so we can apply Chernoff bound on it. Let
  \begin{align*}
    \mu     = \E_{x_1, \ldots, x_n}\Bk{\sum_{j=1}^n z_j} = pn, \quad \text{and} \quad
    \delta  = 1 - \frac{i}{pn},
  \end{align*}
  then the left side of \eqref{eq:cct-2} can be bounded by
  \begin{align*}
    \Pr_{x_1, \ldots, x_n}\Bk{y_i > \frac{i}{n+1} + \frac{\eps}{2}}
    = \Pr_{x_1, \ldots, x_n}\Bk{\sum_{j=1}^n z_j < i}
    = \Pr_{x_1, \ldots, x_n}\Bk{\sum_{j=1}^n z_j < (1-\delta)\mu}
    \le e^{-\delta^2\mu/3}.
  \end{align*}
  Furthermore, since
  \begin{align*}
    \delta^2\mu
    = \frac{(pn - i)^2}{pn}
    = \frac{\bk{\eps n / 2 - \frac{i}{n + 1}}^2}{\frac{ni}{n+1} + \eps n / 2}
    \ge \frac{\bk{2n^{3/4} - 1}^2}{n + n}
    \ge \frac{n^{3/2}}{2n}
    = \frac{\sqrt n}{2}
  \end{align*}
  (here we use the condition $4n^{-1/4} < \eps < 1$), we have
  \begin{align*}
    \Pr_{x_1, \ldots, x_n}\Bk{y_i > \frac{i}{n+1} + \frac{\eps}{2}}
    \le e^{-\delta^2\mu/3}
    \le e^{-\sqrt n / 6}, %
  \end{align*}
  as desired. By the same argument, we can also prove that
  \begin{equation*}
    \Pr_{x_1, \ldots, x_n} \Bk{y_i < \frac{i}{n + 1} - \frac{\eps}{2}} \le e^{-\sqrt n / 6}.
  \end{equation*}
  Hence for any $\eps > 4 n^{-1/4}$, we know that
  \[
  \Pr_{x_1, \ldots, x_n} \Bk{ \abs{y_i - \frac{i}{n + 1}} > \frac{\eps}{2}} \le 2e^{-\sqrt n / 6} .
  \]
  Moreover, since $\eps >  4n^{-1/4}$, we know $\eps > 2(n+1)^{-1}$, and therefore 
  \begin{align*}
    \Pr_{x_1, \ldots, x_n} \Bk{ \abs{y_i - \frac{i}{n}} > \eps} 
    \,\le\, & \Pr_{x_1, \ldots, x_n} \Bk{ \abs{y_i - \frac{i}{n + 1}} > \eps - \frac{i}{n(n+1)}} \\
    \,\le\, & \Pr_{x_1, \ldots, x_n} \Bk{ \abs{y_i - \frac{i}{n + 1}} > \frac{\eps}{2} } 
    \le 2e^{-\sqrt n / 6} .
  \end{align*}
  Take union bound over all $i$, we obtain
  \begin{equation*}
    \Pr_{x_1, \ldots, x_n} \Bk{\abs{y_i - \frac{i}{n}} \le \eps,\, \forall i\in [n]}  \ge 1 - 2ne^{-\sqrt n / 6}\quad \forall \eps \in (4n^{-1/4}, 1). \qedhere
  \end{equation*}
\end{proof}
\subsection{Reduction from Fractional to Random in Vertex-Weighted Scenario}
\label{appendix:reduction}
\newcommand{\bg}{{(\textup{begin})}} %
\newcommand{\ed}{{(\textup{end})}} %
\newcommand{\kt}{(t^*)} %
\newcommand{\kkt}{(t^{*})} %
\newcommand{\pbx}{x} %
\newcommand{\apbx}{z}

This subsection provides the proof of \Cref{thm:rankingbalance}. This lemma shows that for a fixed perturbation function $f$, \prank does not have a better competitive ratio than \pbalance in the vertex-weighted scenario. Therefore, an upper bound of the competitive ratio for \pbalance can also be used to bound the competitive ratio for \pbalance. Let us recall the formal description.

\thmrankingbalance*

To prove the lemma,  we fix an arbitrary instance $G=(L\cup R, E,w)$ and a perturbation function $f$, and we would like to construct another instance $G'$ so that 
\[
\Gamma_{\textsf{balance}}(G) \approx \Gamma_{\textsf{ranking}}(G'),
\]
where $\Gamma_{\textsf{balance}}(G)$ denotes the competitive ratio of \pbalance for $G$ and $\Gamma_{\textsf{ranking}}(G')$ denotes the competitive ratio of \prank for $G'$.

This is done in two steps.

The first step is to argue that the worst-case performance of \prank can be compared to a special algorithm class called: \apbalance. (We will formally define \apbalance later.)

The second step is to construct an instance $G'$ that
\[
\Gamma_{\textsf{balance}}(G) \approx \Gamma_{\textsf{approximate balance}}(G').
\]
Here, $\Gamma_{\textsf{approximate balance}}(G')$ is the competitive ratio of \apbalance for $G'$. It means the property holds for every algorithm in the class \apbalance. We proceed by first formally defining \apbalance.

\begin{definition}
  \label{def:apbalance}
  An algorithm $A$ is $\eps$-\apbalance with perturbation function $f$ if at the arrival time of an online vertex $v$, the algorithm $A$ continuously matches $v$ to an offline vertex $u$ in its neighborhood which satisfies the following condition:
  \begin{equation}
    \label{eq:awf-cond}
    w_u f(x_u - \eps) \ge w_a f(x_a + \eps), \quad \forall a \in N(v) \text{ with } x_a < 1.
  \end{equation}
  That is, $u$ is approximately the offline vertex with the highest perturbed weight. 
\end{definition}

Notice that when $\eps = 0$, this definition is exactly \pbalance. When $\eps>0$, \apbalance becomes a class of algorithms. Any algorithm who chooses offline vertex $u$ following the condition \eqref{eq:awf-cond} can be called an $\epsilon$-\apbalance. For rigorous issue, we also extend the definition of $f$ by $f(x) = +\infty,\,\forall x< 0$ and $f(x) = 0,\, \forall x > 1$ in case of $x_u < \eps$ or $x_a > 1 - \eps$.

\begin{lemma}
  \label{lemma:hw3lm1}
  Let $f$ be any fixed perturbation function. For any instance $G$ of fractional vertex-weighted online bipartite matching and $\delta > 0$, there exists $\eps > 0$ such that
  \begin{equation}
    \label{hw3lm1}
    \textup{APB}_{\eps, f}(G) \le \textup{PB}_{f}(G) + \delta \cdot \textup{OPT}(G),
  \end{equation}
  where $\textup{APB}_{\eps, f}(G)$ is the maximum possible gain of any $\eps$-\apbalance algorithm acting on $G$ with $f$, $\textup{PB}_{f}(G)$ is the gain of \pbalance on $G$ with $f$, and $\textup{OPT}(G)$ is the optimal gain in this instance. 
\end{lemma}
\begin{proof}
  We first introduce some notations. Clearly we only need to consider the case that $\opt(G)>0$. Fix an arbitrary instance $G$ and perturbation function $f$, and use $\pbx_u(t)$ and $\apbx_u(t)$ to denote the matched portion of $u$ at time $t$ in \pbalance and \apbalance, respectively. 
  Our goal is to compare the performance of \pbalance and \apbalance on this instance $G$.

  Without loss of generality, we assume $f$ is left-continuous. Moreover, note that $\apbx_u$ and $\pbx_u$ are both continuous functions of $t$. We also assume at least $m$ offline vertices in $G$ with zero weight, fully connected with online vertices. Here $m = |R|$ is the number of online vertices. Under this assumption, every online vertex must be matched completely, so at any time $t$, we have 
  \begin{equation}
    \label{hw3lm1eq1}
    \sum_{u \in L} \apbx_u(t) = \sum_{u \in L} \pbx_u(t).
  \end{equation}
   Let $t$ be the time just before $v_i$'s arrival. We define
  \begin{equation*}
    \Delta_i = \max_{u \in L} (\apbx_u(t) - \pbx_u(t)).
  \end{equation*}
  Combining it with \eqref{hw3lm1eq1}, we immediately have
  \begin{equation}
    \label{hw3lm1eq3}
    \pbx_u(t) - \apbx_u(t) \le (n-1) \Delta_i. 
  \end{equation}

Notice that $\Delta_{m+1}$ shows the difference in the final performance between an \apbalance algorithm and \pbalance. Our next goal is to derive a recurrence relation between $\Delta_i$ and $\Delta_{i+1}$ and use this to bound $\Delta_{m+1}$.
We fix an arbitrary $i \le m$ and focus on the process of continuously matching an online vertex $v = v_i$. 
We use $\pbx_u\bg$ and $\apbx_u\bg$ to denote the matched portion of offline vertex $u$ just before $v$ arrives in \pbalance and \apbalance, respectively. We also use $\pbx_u\ed$ and $\apbx_u\ed$ to denote the matched portion just after $v$ is completely matched.

  We define a critical moment $t^*$. Consider the following set:
  \[
    T = \left\{ t : \text{At time }t\text{, \pbalance matches }v\text{ to a node }a\text{ with }\apbx_{a}(t^*) \ge \pbx_{a}\ed + \eps\right\}.
  \]
  Intuitively, we want to define $t^* = \inf T$ so as to bound $\apbx_u \ed - \pbx_u \ed $ by bounding both $\apbx_u \ed - \apbx_u \kt$ and $\apbx_u \kt - \pbx_u \ed$. But for technical reason, we need $t^* \in T$, hence we need to allow $t^* - \inf T$ to be a non-zero but sufficiently small number. Specifically, if $T$ is non-empty, we let $t^*$ be an element of $T$ such that at most $\eps$ portion of $v$ arrived in $[\inf T, t^*)$, otherwise we just let $t^*$ be the moment that $v$ is just complete matched.
  Then, we prove that for all $u$:
  \begin{equation}
    \label{eq:first_piece_bound}
    \apbx_u\kt - \pbx_u\ed \le \max\BK{2\eps, \Delta_i}, \quad \forall u \in L.
    \tag{P1}
  \end{equation}
  
  Obviously, for those $u$ such that $\apbx_u\kt = \apbx_u\bg$, we have
  \[\apbx_u\kt - \pbx_u\ed = \apbx_u\bg - \pbx_u\ed \le \apbx_u\bg - \pbx_u \bg \le \Delta_i.\]
  
  Otherwise, $u$ must have been chosen by $v$ in \apbalance before moment $t^*$. Let $t'$ be the last time before $t^*$ that $u$ is chosen, i.e., 
  $t' = \sup \left\{ t < t^*: u\text{ is chosen at moment } t \right\}$. 
  Then by the definition of $t^*$, $\apbx_{u}(t') \le \pbx_{u}\ed + 2\eps$ must hold since $\apbx_u(t') \le \apbx_u\kt \le \apbx_u(\inf T) + \eps$ and each moment $t<\inf T$ such that $u$ is chosen at moment $t$ must satisfy $\apbx_{u}(t) < \pbx_{u}\ed + \eps$. And thus \[\apbx_u\kt - \pbx_u\ed = \apbx_u(t') - \pbx_u\ed \le \eps.\]

  Thus far, we conclude \Cref{eq:first_piece_bound}. 
  Notice that when $T$ is empty, \eqref{eq:first_piece_bound} immediately implies a recurrence relation \[\Delta_{i+1} \le \max \{2\eps, \Delta_i\} \le 4n^2( \Delta_i + \eps).\] 
  Hence we only need consider the case that $T \neq \emptyset$, which means $t^* \in T$ by definition, and at time $t^{*}$, \pbalance matches $v$ to a node $a$ with $\apbx_{a}(t^{*}) \ge \pbx_{a}\ed + \eps$. In this case, we will upper bound $\apbx_u\ed - \apbx_u\kt$ as follows:
  \begin{equation}
    \label{eqn:second_piece_main}
    \apbx_u\ed - \apbx_u\kt \le n^2 \Delta_i + 2n \eps.
    \tag{P2}
  \end{equation}
  Suppose \Cref{eqn:second_piece_main} hold. Combining with \Cref{eq:first_piece_bound}, we will have
  \begin{equation*}
    \apbx_u\ed - \pbx_u \ed = (\apbx_u\ed - \apbx_u \kt) + (\apbx_u\kt - \pbx_u\ed) \le \max\BK{2\eps, \Delta_i} + n^2\Delta_i + 2n\eps\le 4n^2(\Delta_i + \eps), 
  \end{equation*}
  which means $\Delta_{i+1} \le 4n^2 (\Delta_i + \eps)$. Moreover, by the initial value $\Delta_1 = 0$, we can see
  \begin{equation*}
    \Delta_{m+1} \le m(4n^2)^m\eps.
  \end{equation*}
  by a simple mathematical induction, where $\Delta_{m+1}$ represents the maximum $\apbx_u - \pbx_u$ after all online vertices are matched. And finally,
  \[\textup{APB}_{\eps, f}(G) - \textup{PB}_{f}(G) = \sum_{u \in L} w_u \bk{\apbx_u - \pbx_u} \le \sum_{u \in L} w_u \Delta_{m+1} \le m(4n^2)^m\sum_{u \in L} w_u \eps.\]
  Let
  $\eps = (\delta \cdot \opt(G)) / \bk{ m(4n^2)^m\sum_{u \in L} w_u}$
  and we will get the desired statement.
  
  \subparagraph*{Proof of \Cref{eqn:second_piece_main}.}Finally, let us complete the proof of \Cref{eqn:second_piece_main}.
  At first, by \Cref{hw3lm1eq1}, we have 
  \begin{equation}
    \label{eq:idea_of_second_piece_bound}
    \apbx_u\ed - \apbx_u\kt \le \sum_{u'\in L} \bk{\apbx_{u'}\ed - \apbx_{u'}\kt} = \sum_{u'\in L} \bk{\pbx_{u'}\ed - \apbx_{u'}\kt}.
  \end{equation}
  It remains to show  
  \begin{equation}
    \label{hw3lm1eq4}
    \pbx_u\ed - \apbx_u\kt \le n\Delta_i + 2\eps, \quad \forall u \in L.
  \end{equation}

  We first rule out some simple cases. If $\apbx_u\kkt = 1$, i.e., at moment $t^{*}$, all budget of $u$ has been used up, then \eqref{hw3lm1eq4} is clear. So we only need to consider the case that $\apbx_u\kkt < 1$. In this case, we must have
  \begin{equation}
    \label{hw3case2eq1}
    w_{a} f(\apbx_{a}\kkt - \eps) \ge w_u f(\apbx_u\kkt + \eps),
  \end{equation}
  since at moment $t^{*}$, \apbalance matches $a$ to $v$.

  Moreover, we can also ruin out the case $\pbx_u \bg = \pbx_u \ed$, which means $u$ has never been chosen by $v$ in \pbalance, then by \eqref{hw3lm1eq3}, we already have \eqref{hw3case2eq1}.
  \[ \pbx_u\ed = \pbx_u\bg \le \apbx_u\bg + n\Delta_i \le \apbx_u\kkt + n\Delta_i.  \]
  
  So we only need to consider the case under $\pbx_u \bg < \pbx_u \ed$ and $z_u(t^*)<1$. In this case, $u$ must be chosen at some moment, and $a$ is never sutured. It is from the definition of $a$, we have $\pbx_a\ed\leq \apbx_a(t^*) - \eps \leq 1 - \eps$. In this case, the perturbed weight of $u$ is the largest at some moment. After that, by the balanced idea of \pbalance, it will keep being the largest among all unsaturated offline vertices. Therefore, we plan to prove the following observation:
  \begin{equation*}
    w_u f(\pbx_u \ed) \ge w_a f(\pbx_a \ed).
  \end{equation*}
  Combining this property with \eqref{hw3case2eq1}, we have 
  \[
    w_u f(\pbx_u\ed) \ge w_a f(\pbx_a\ed) \ge w_a f(\apbx_a\kkt - \eps) \ge w_u f(\apbx_u\kkt + \eps).
  \]
  By the monotony of $f$, we know $\pbx_u \ed \le \apbx_u\kkt + \eps$, and thus \eqref{hw3lm1eq4} has been verified. As previous discussion goes, plugging \eqref{hw3lm1eq4} into \eqref{eq:idea_of_second_piece_bound}, we can conclude \Cref{eqn:second_piece_main}.
  \begin{equation}
    \apbx_u\ed - \apbx_u\kt \le n^2 \Delta_i + n \eps \leq n^2 \Delta_i + 2n \eps.
    \tag{\ref{eqn:second_piece_main}}
  \end{equation}

  However, if $f$ is only right continuous, \pbalance may not be able to keep all perturbed weight the same at every moment. We present the following proof for the rigorous issue. We prove a slightly weaker observation
  \begin{equation*}
    w_u f(\pbx_u \ed - \eps) \ge w_a f(\pbx_a \ed).
  \end{equation*}
  We remark that it let us suffer another $\epsilon$ factor, which is why we have $2\epsilon$ in \eqref{hw3lm1eq4}.
  In particular, by $\pbx_u \bg < \pbx_u \ed$, we know $u$ has been chosen by \pbalance at some moment when $u$ has the highest perturbed weight among all available nodes. Hence we can define
  \[
    t' = \sup \BK{t : w_u f(\pbx_u(t)) \ge w_a f(\pbx_a(t))}.
  \]
  (Note that here we reuse the notation $t'$ to represent a new quantity.) From the definition of $t'$ we can derive that $\pbx_u(t') = \pbx_u \ed$ and $w_u f(\pbx_u(t') - \eps) \ge w_a f(\pbx_a(t'))$. The first one is because $a$ is never used up, and the gain of choosing $u$ is never more than the gain of choosing $a$ after moment $t'$, so $u$ is never been chosen. To see the second one, 
  by the definition of $t'$ we can see that there is a $t\in (t' - \eps , t']$ such that $w_u f(\pbx_u(t)) \ge w_a f(\pbx_a(t))$, 
  and hence we get $w_u f(\pbx_u(t') - \eps) \ge w_u f(\pbx_u(t)) \ge w_a f(\pbx_a(t)) \ge w_a f(\pbx_a(t))$. Combining these two properties of $t'$, we can justify our observation:
  \[
    w_u f(\pbx_u \ed - \eps) = w_u f(\pbx_u(t') - \eps) \ge w_a f(\pbx_a(t')) \ge w_a f(\pbx_a \ed).
  \]
  Combining this observation with \eqref{hw3case2eq1}, we have 
  \[
    w_u f(\pbx_u\ed - \eps) \ge w_a f(\pbx_a\ed) \ge w_a f(\apbx_a\kkt - \eps) \ge w_u f(\apbx_u\kkt + \eps).
  \]
  By the monotony of $f$, we prove $\pbx_u \ed \le \apbx_u\kkt + 2\eps$ as \eqref{hw3lm1eq4} and conclude \eqref{eqn:second_piece_main}.
\end{proof}

\begin{lemma}
  \label{lemma:hw3lm2}
  Let $f$ be any fixed perturbation function. For any instance $G$ of fractional vertex-weighted online bipartite matching and positive constants $\delta$ and $\eps$, there exists an instance $G'$ such that $\textup{OPT}(G') = \textup{OPT}(G)$ and
  \begin{equation}
    \label{hw3lm2}
    \textup{PR}_{f}(G') \le \textup{APB}_{\eps, f}(G) + \delta \cdot \textup{OPT}(G),
  \end{equation}
  where $\textup{PR}_f$ is the gain of \prank with perturbation function $f$, and $\textup{APB}_{\eps, f}(G)$ is the maximum possible gain of any $\eps$-\apbalance acting on $G$ with the same perturbation function $f$.
\end{lemma}
\begin{proof}
  We construct $G'$ in the following way. Say $G = (L \cup R, E)$ and $G' = (L' \cup R', E')$. Let $N$ be a large positive integer. For every offline vertex $u \in L$, construct $N$ offline vertices $u^{(1)}, u^{(2)}, \ldots, u^{(N)}$ in $L'$, with weight $w_{u^{(i)}} = w_u/N$. For every online vertex $v \in R$, construct $N$ online vertices $v^{(1)}, \ldots, v^{(N)}$ in $R'$. They arrives consecutively at the same time when $v$ arrives in $G$. For each edge $(u, v) \in E$, we connect all vertex pairs $(u^{(i)}, v^{(j)})$ for $i, j \in [N]$. It is clear that $\textup{OPT}(G) = \textup{OPT}(G')$, and we simply write it as $\opt$ for short.

  At the beginning of the \prank algorithm, we assign a rank $y_{u^{(i)}}$ to every vertex $u^{(i)} \in L'$. Without loss of generality, we assume $y_{u^{(1)}} \le y_{u^{(2)}} \le \cdots \le y_{u^{(N)}}$. By  \Cref{cctlemma}, there exists a sufficiently large number $N$ such that the probability of the following event is at least $1-\delta$.
  \begin{equation}
    \label{hw3lm2cond}
    \abs{y_{u^{(i)}} - \frac{i}{N}} \le \frac{\eps}{2}, \quad \forall u \in L \text{ and } i \in [N].
  \end{equation}
  
  Let us choose a sufficient large $N>2/\eps$, with the probability at least $1-\delta$, we have \Cref{hw3lm2cond} and we will prove  $\textup{PR}_f(G') \le \textup{APB}_{\eps, f}(G)$. On the other hand, if \Cref{hw3lm2cond} fails, we use $\opt$ to upper bound $\textup{PR}_f(G')$. To sum up, we will have
  \begin{equation*}
    \textup{PR}_f(G') \le (1-\delta) \textup{APB}_{\eps, f}(G) + \delta \cdot \textup{OPT} \le \textup{APB}_{\eps, f}(G) + \delta \cdot \textup{OPT}.
  \end{equation*}

  It remains to prove $\textup{PR}_f(G') \le \textup{APB}_{\eps, f}(G)$ when we have \Cref{hw3lm2cond}. We regard \prank on $G'$ as a process of fractional matching on $G$: when \prank matches $v^{(i)}$ to $u^{(j)}$, obtaining a gain of $w_{u} / N$, it is equivalent to matching $1/N$ portion of $v$ to $u$, obtaining the same gain. \prank algorithm can be restated under this view:
  \begin{itemize}
    \item For each $u\in L$, assign ranks $y_{u^{(1)}} \le y_{u^{(2)}} \le \cdots \le y_{u^{(N)}}$.
    \item When an online vertex $v$ arrives, make $N$ decisions. On each decision, choose the $u$ with largest $w_u f(y_{u^{\bk{x_u N + 1}}})$, match $v$ to $u$ for $1/N$ portion.
  \end{itemize}
  Suppose \eqref{hw3lm2cond} holds. On each decision moment of an online node $v$, if \prank algorithm chooses $a \in L$ and there is $x_a$ portion of the budget of $a$ has been used, then we know it is the $x_a N$-th copy of $a$ that is being matched (note that $x_aN$ is an integer at the moment when $v$ is just about to arrive), and hence for any other choice $u\in N(v)$ with used budget $x_u < 1$, we have 
  \begin{equation*}
    w_a f(x_a - \eps / 2) \ge w_a f(y_{a^{\bk{x_a N}}}) \ge w_u f(y_{u^{\bk{x_u N}}}) \ge w_u f(x_u + \eps / 2).
  \end{equation*}
  Here we use the fact that $y_{a^{\bk{x_a N}}} \ge x_a - \eps/2$ and $y_{u^{\bk{x_u N}}} \le x_u + \eps / 2$ by \eqref{hw3lm2cond}.
  After a decision is made, the next $1/N$ portion is matched to the chosen vertex. At any time in the next $1/N$ portion, we have
  \begin{equation*}
    w_a f(x_a - \eps / 2 - 1/N) \ge w_u f(x_u + \eps / 2), \quad \forall u \in N(v) \text{ with } x_u < 1.
  \end{equation*}
  Since $1/N < \eps / 2$, we know the condition of $\eps$-\apbalance $w_a f(x_a - \eps) \ge w_u f(x_u + \eps)$ is satisfied, so the gain of \prank is bounded by the best possible gain of \apbalance, i.e., $\textup{PR}_f(G') \le \textup{APB}_{\eps, f}(G)$.
\end{proof}

Finally, combining \Cref{lemma:hw3lm1} and \Cref{lemma:hw3lm2}, it is ready to prove \Cref{thm:rankingbalance}.

\begin{proof}[Proof of \Cref{thm:rankingbalance}]
  For any $\delta > 0$, we prove that $\Gamma_{\textsf{ranking}} \le \Gamma_{\textsf{balance}} + \delta$. First, there exists an instance $G$ such that $\text{PB}_f(G) \le (\Gamma_{\textsf{balance}} + \delta/3) \textup{OPT}(G)$, where $\text{PB}$ is the gain of \pbalance. 
  
  By \Cref{lemma:hw3lm1}, we know there exists some $\eps > 0$ such that $\text{APB}_{\eps, f}(G) \le \text{PB}_f(G) + (\delta/3) \text{OPT}(G) \le (\Gamma_{\textsf{balance}} + 2\delta / 3) \textup{OPT}(G)$. By \Cref{lemma:hw3lm2}, we know there exists some instance $G'$ with the same optimal gain, such that $\text{PR}_f(G') \le \text{APB}_{\eps, f}(G) + (\delta / 3)\text{OPT}(G) \le (\Gamma_{\textsf{balance}} + \delta) \text{OPT}(G')$.
  
  This shows that $\Gamma_{\textsf{ranking}} \le \Gamma_{\textsf{balance}} + \delta$. Let $\delta \to 0$, we have $\Gamma_{\textsf{ranking}} \le \Gamma_{\textsf{balance}}$.
\end{proof}

\section{Proof of \Cref{lem:mathematical_fact}}
\label{appendix:mathematical}
We prove the following mathematical fact.
\lemratio*

\begin{proof}
  We prove the lemma by contradiction, and we assume $\Gamma = 1 - 1/e - \gamma$, where $\gamma = 0.0003$.

  Let $r$ be the root of $1 - e^{r}(1 - \Gamma)$, i.e., $r = -\ln (1 - \gamma) \approx 0.999185$.
  We choose parameters $\delta = 0.05$ and $\beta^* \approx 0.009615$, where $\beta^*$ is the unique root of
  \[
  g(\beta) = (\delta + e\gamma + \delta e \gamma) (e^{\beta-1} - e^{-1}) + \gamma - \beta\delta.
  \]

  We can also see that $\beta^* < \Gamma$.

  Now for any $\alpha \in (0,r)$, we choose an $\beta$ to simplify \eqref{ljxeq2}. Since $\alpha < r$, we can ensure that $1 - e^\alpha(1 - \Gamma)>0$. Let
  \begin{equation*}
    \beta = \inf \left\{ x \in [0,\alpha] : \frac{f(x)}{1 - e^x(1 - \Gamma)} \le (1 + \delta) \frac{f(\alpha)}{1 - e^\alpha(1 - \Gamma)} \right\}.
  \end{equation*}
  The existence of $\beta$ can be seen by noticing $\alpha$ is an element of the set in the right hand side. 
  
  The choice of $\beta$ leads to two important properties:
  \begin{gather}
    \label{eqn:calculation_first_pro}
    \frac{f(\alpha)}{f(\beta)} \ge \frac{1 - e^\alpha(1 - \Gamma)}{(1 + \delta)\Gamma},
    \tag{P1}\\
    \label{eqn:calculation_second_pro}
    \beta\le \beta^*.
    \tag{P2}
  \end{gather}
  \eqref{eqn:calculation_first_pro} comes from the definition of $\beta$. By the definition, for any $\epsilon >0$, there exists $x\in (\beta, \beta + \epsilon)$ such that
  \[
    \frac{f(x)}{1 - e^x(1 - \Gamma)} \le (1 + \delta) \frac{f(\alpha)}{1 - e^\alpha(1 - \Gamma)}.
  \]
  Since $f$ is right continuous, let $\epsilon \to 0$ and we get
  \[
    \frac{f(\beta)}{1 - e^\beta(1 - \Gamma)} \le (1 + \delta) \frac{f(\alpha)}{1 - e^\alpha(1 - \Gamma)},
  \]
  which means that
  \[
    \frac{f(\alpha)}{f(\beta)} \ge \frac{1 - e^\alpha(1 - \Gamma)}{(1 + \delta)(1 - e^\beta(1 - \Gamma))} \ge \frac{1 - e^\alpha(1 - \Gamma)}{(1 + \delta)\Gamma}.
  \]

  Next, we prove \eqref{eqn:calculation_second_pro}. Similar to \Cref{thm:vertexweight}, we split \eqref{ljxeq1} into two equations, and due to the homogeneity, we can assume $M = 1$, i.e.,
  \begin{gather}
    f(\alpha) \ge 1 - e^{\alpha} (1 - \Gamma), \label{ljxeq3} \\
    \int_{0}^{\beta} f(x) \dx \le  \beta+1-e^{\beta-1} - \Gamma. \label{ljxeq4}
  \end{gather}
  By the definition of $\beta$ and \eqref{ljxeq3}, we know for any $x\in (0,\beta)$,
  \[
    f(x) \ge (1 + \delta) \frac{f(\alpha)}{1 - e^\alpha(1 - \Gamma)} (1 - e^x(1 - \Gamma)) \ge (1 + \delta)(1 - e^x(1 - \Gamma)),
  \]
  hence
  \[
    \int_{0}^{\beta} f(x) \dx \ge (1 + \delta) \int_0^\beta (1 - e^x(1 - \Gamma)) \dx = (1+\delta)(\beta - (e^\beta - 1)(1 - \Gamma)).
  \]
  Comparing this lower bound to \eqref{ljxeq4}, we get
  \begin{equation*}
    \beta+1-e^{\beta-1} - \Gamma \ge (1+\delta)(\beta - (e^\beta - 1)(1 - \Gamma)).
  \end{equation*}
  After substituting $\gamma = 1 - 1/e - \Gamma$ into it and simplifying, we get
  \begin{equation}
    (\delta + e\gamma + \delta e \gamma) (e^{\beta-1} - e^{-1}) + \gamma - \beta\delta \ge 0. \label{ljxeq5}
  \end{equation}
  We regard the left hand side of \eqref{ljxeq5} to be a function of $\beta$ and denote it by $g(\beta)$. Then clearly
  \[
    g''(\beta) = (\delta + e\gamma + \delta e \gamma) e^{\beta-1} \ge 0,
  \]
  and hence $g$ is a convex function. Moreover, since $g(0) = \gamma > 0$ and $g(1) = (\delta + e\gamma + \delta e \gamma) (1 - e^{-1}) + \gamma - \delta \approx -0.017553 < 0$, $g$ has a unique root over $[0,1]$, which is just $\beta^*$  we have defined. Hence by \eqref{ljxeq5} we can get $\beta \le \beta^*$.

  By using the two important properties \eqref{eqn:calculation_first_pro} and \eqref{eqn:calculation_second_pro}, \Cref{ljxeq2} becomes
  \begin{align*}
    (1-\Gamma) f(\alpha) & \ge (\Gamma - \alpha)  \int_0^\alpha f(x) \dx + (\Gamma - \beta) \int_{\alpha}^1 \min \left\{f(x),\frac{f(\alpha)}{f(\beta)}f(x-\alpha)\right\} \dx                                         \\
                         & \, \ge (\Gamma - \alpha)  \int_0^\alpha f(x) \dx + (\Gamma - \beta^*) \int_{\alpha}^{1} \min \left\{f(x),\frac{1 - e^\alpha(1 - \Gamma)}{(1 + \delta)\Gamma}f(x-\alpha)\right\} \dx         \\
                         & \, \ge (\Gamma - \alpha)  \int_0^\alpha (1 - e^{x} (1 - \Gamma)) \dx                                                                                                                        \\
                         &\quad\   + (\Gamma - \beta^*) \int_{\alpha}^{r} \min \left\{(1 - e^{x} (1 - \Gamma)),\frac{1 - e^\alpha(1 - \Gamma)}{(1 + \delta)\Gamma}(1 - e^{x - \alpha} (1 - \Gamma))\right\} \dx.
  \end{align*}
  Here the first inequality is just \eqref{ljxeq2}, the second inequality is by the property of $\beta$, and in the third inequality, we just throw a part of second integral (in range $(r,1)$) away and use \eqref{ljxeq3}. Now the final lower bound of $(1 - \Gamma)f(\alpha)$ is a function of $\alpha$ which is independent to $f$, and we denote it by $(1-\Gamma)h(\alpha)$. So we get
  \[
    f(\alpha) \ge h(\alpha), \quad \forall \alpha \in (0,r).
  \]
  Combining it with \eqref{ljxeq3}, we get
  \[
    f(\alpha) \ge \max\{h(\alpha), 1 - e^{\alpha} (1 - \Gamma)\}, \quad \forall \alpha \in (0,r).
  \]
  We integrate it from $0$ to $r$ numerically.
  \footnote{The code is available at 
  \url{https://github.com/orbitingflea/perturbation-function}.
  }
  \[
    \int_0^r f(x) \dx \ge \int_0^r \max\{h(\alpha), 1 - e^{\alpha} (1 - \Gamma)\} \approx 0.368282 > 0.3682.
  \]
  There is a contradiction with \eqref{ljxeq4}, because 
  \[
    \int_0^r f(x) \dx \le  r+1-e^{r-1} - \Gamma = 0.368179 < 0.3682
  \]
  Hence we must have $\Gamma < 1 - 1/e - 0.0003$.
\end{proof}

\end{document}